\RequirePackage[T1]{fontenc}
\documentclass[conference,onecolumn,10pt,letterpaper]{IEEEtran}
\usepackage{fontspec}
\usepackage{amsmath,amssymb}
\usepackage{unicode-math}
\usepackage{graphicx,booktabs,array,longtable}
\usepackage{enumitem}
\usepackage{xcolor}
\usepackage{tikz}
\usepackage{hyperref}
\hypersetup{hidelinks}
\usepackage{placeins,needspace}

\newcommand{\F}{\mathbb F}
\newcommand{\Z}{\mathbb Z}

\newcommand{\wt}{\operatorname{wt}}
\newcommand{\rank}{\operatorname{rank}}
\newcommand{\dmin}{d_{\min}}
\newtheorem{theorem}{Theorem}
\newtheorem{lemma}[theorem]{Lemma}
\newtheorem{corollary}[theorem]{Corollary}
\newtheorem{definition}[theorem]{Definition}

\setlist{nosep}
\title{CPM-LDPC Codes\\Attaining the Minimum-Distance Bound}
\author{\IEEEauthorblockN{Kenta Kasai}
\IEEEauthorblockA{Department of Information and Communications Engineering\\
School of Engineering, Institute of Science Tokyo}}

\hypersetup{pdfauthor={Kenta Kasai},pdftitle={CPM-LDPC Codes Attaining the Minimum-Distance Bound}}
\begin{document}
\maketitle
\begin{abstract}
We study binary quasi-cyclic LDPC codes whose parity-check matrices are full arrays of single circulant permutation matrices (CPMs), referred to here as CPM-LDPC codes. Their minimum distance is at most $(J+1)!$, where $J$ is the column weight. For every fixed pair of column and row weights $2\le J<L$, we show that this bound is attained for all sufficiently large integer lift sizes. First, we give one integer exponent matrix independent of the lift size $P$. Second, we show that independent uniform exponent choices attain the bound with probability $1-O_{J,L}(P^{-1})$. Both proofs use cycle conditions required by low-weight codewords and a lower bound on the number of terms in vectors satisfying polynomial check equations. Neither construction requires $P$ to be prime. We also give small-lift arrays attaining the bound 24 for $J=3$, $L=4,\ldots,8$, and arrays with distance at least 28 for $J=4$, $L=5,\ldots,8$, together with computational distance verification.
\end{abstract}

\section{Introduction}
We use CPM-LDPC to denote binary low-density parity-check (LDPC) codes whose parity-check matrices have a single circulant permutation matrix (CPM) in every block. These are a class of quasi-cyclic LDPC (QC-LDPC) codes and are also called QC-LDPC codes. The name CPM-LDPC makes their construction from circulant permutation blocks explicit. Each CPM is specified by one cyclic shift, so a large check matrix is described by few parameters. The cyclic structure also supports encoding and decoder implementation. Tanner et al.~\cite{tanner2004} constructed QC-LDPC codes and their convolutional counterparts from circulant matrices and studied their distance, Tanner-graph cycles, and encoding methods.

A $J\times L$ array of $P\times P$ circulant permutation matrices has column weight $J$, row weight $L$, and length $LP$. We call $P$ the lift size. This format imposes a restriction: increasing $P$ cannot raise the minimum distance above $(J+1)!$. The factorial bound originates with MacKay and Davey~\cite{mackay2001}; Smarandache and Vontobel~\cite[Corollary~9]{sv2011} derive it as a special case of permanent-based bounds that apply to more general arrays. For fixed $J,L$, a basic distance-design question is therefore whether this bound can be attained.

Exponent design often uses conditions that exclude short cycles. Fossorier~\cite{fossorier2004} gave girth conditions for arrays of circulant permutation matrices in terms of sums of exponent differences along cycles, as well as a necessary condition for attaining the factorial distance bound. Smarandache and Mitchell~\cite{sm2022} developed conditions and algorithms for constructing QC-LDPC codes with prescribed girth. Bocharova et al.~\cite{bocharova2011} combined searches based on cycle lengths with minimum-distance computations and gave examples attaining the distance bound. Attainment itself is therefore already known for selected parameters. Excluding short cycles alone, however, does not establish the absence of every codeword below the distance bound.

We ask whether, for every fixed $2\le J<L$, one integer exponent matrix independent of $P$ attains the factorial bound for all sufficiently large integers $P$. Section~\ref{sec:explicit} answers this affirmatively with an explicit construction whose exponents are powers of a fixed integer. Section~\ref{sec:random} further shows that, when the exponents are independent and uniform, the probability of falling below the bound is at most a constant times $P^{-1}$. Both results apply to composite lift sizes as well as primes.

The proofs use conditions on the exponents that are necessary for a low-weight codeword to exist. Section~\ref{sec:preliminaries} defines a graph by pairing the 1-positions of a codeword at each check and obtains congruence conditions from its cycles. Sums of exponents along cycles are familiar in CPM-LDPC analysis; here they are applied to the entire support of a codeword. If each cycle sum, viewed as an integer linear combination of the exponents, has all coefficients zero, we can represent each 1-position by a monomial and rewrite the parity checks as polynomial identities. The contributions from two vertices paired at a check become the same monomial and cancel in that check equation. Within each component, however, distinct positions in the same block column give distinct monomials, so no terms cancel and the codeword weight is preserved as the total number of terms in the components. This correspondence allows us to use a theorem of Draisma, Kahle, and Wiersig on polynomial term counts~\cite[Theorem~3.1 and Corollary~3.3]{dkw2023} to show that a codeword below the bound requires at least one cycle congruence with a nonzero integer coefficient vector. The explicit construction excludes that condition, while the random construction bounds its probability. The algebraic derivation of the term-count bound is given in the appendix.

The explicit construction uses large exponents and a large sufficient lift size, and the random bound has a large constant. Examples at small $P$ therefore complement these general guarantees. Section~\ref{sec:examples} presents, in a common format, examples attaining the bound 24 for $J=3$ and examples with distance at least 28 for $J=4$. The latter demonstrate distances unavailable with $J=3$ in this array format; they do not attain the $J=4$ bound of 120. The reported distances are verified from the completed search scope for each exponent array.

The array format matters. Mitchell et al.~\cite{mitchell2014} use a two-step construction that first enlarges a small graph and then applies a circulant lift, improving distance and girth relative to a direct circulant lift of the original graph. We retain a fixed-size array in which all $JL$ blocks are single circulant permutation matrices and study attainment within that format. The factorial bound is not asserted for arbitrary QC-LDPC codes specified only by their column and row weights. For fixed $J,L$, the guaranteed distance is constant, so the relative minimum distance $\dmin/(LP)$ tends to zero as $P$ grows.

This study grew out of the question of improving the minimum distance of pair-partition (PP) quantum codes built from CPMs~\cite{okada2026pp} by choosing their exponent arrays or lift sizes. That question led us to study attainment of the distance bound for a single parity-check matrix. Appendix~\ref{sec:quantum} discusses the additional analysis needed for an extension to quantum codes.

\section{Preliminaries}
\label{sec:preliminaries}
Attaining the distance bound requires excluding every nonzero codeword of smaller weight. To express the necessary conditions in terms of the exponent matrix, we first define the code and its positions, then construct a graph from a codeword. The cycle equations and connectivity for minimum-weight codewords provide a common starting point for the two constructions.

\subsection{Code definition}
\label{sec:code-definition}
A circulant permutation matrix shifts $P$ positions by a fixed amount. For example, with $P=5$, a shift by 2 sends positions $0,1,2,3,4$ to $2,3,4,0,1$, respectively. The exponent matrix $E$ lists these shift amounts $e_{j\ell}$. Positions are identified by their remainders modulo $P$, so adding $P$ to an exponent leaves the matrix unchanged.

Let $2\le J<L$ and $P\ge2$ be integers. For an exponent matrix $E=(e_{j\ell})$, let $H_P(E)$ be the binary parity-check matrix whose Tanner graph joins variable node $(\ell,t)$ to check node $(j,t+e_{j\ell})$. Variable and check nodes correspond to columns and rows, respectively. The indices satisfy $0\le j<J$, $0\le\ell<L$, and $t\in\Z/P\Z$, where $\Z/P\Z$ denotes positions $0,1,\ldots,P-1$, with position sums reduced modulo $P$. Every block is a single circulant permutation matrix, with no zero blocks. Define
\begin{equation}
 C_P(E)=\ker_{\F_2}H_P(E),\qquad
 \dmin(C_P(E))=\min_{0\ne c\in C_P(E)}\wt(c).
 \label{eq:code}
\end{equation}
In the binary field $\F_2$, the elements are 0 and 1, with $1+1=0$. The notation $\ker_{\F_2}H_P(E)$ denotes all binary vectors $c$ satisfying $H_P(E)c=0$: each check must contain an even number of 1s. The support is the set of 1-positions, and the weight $\wt(c)$ counts them. The minimum distance is the smallest weight of a codeword other than the all-zero word. The column weight is $J$, the row weight is $L$, and the length is $LP$. Summing the $P$ rows in any block row gives the all-one vector, so every codeword has even weight. The $J-1$ independent dependencies between these block-row sums give
\begin{equation}
 k\ge (L-J)P+J-1.
 \label{eq:dimension}
\end{equation}
Indeed, summing all check equations in one block row includes each variable exactly once. The sum of all codeword coordinates is therefore zero, so the weight is even. Also, adding the row sum of block row 0 to that of block row $j$ gives the zero vector for each $j=1,\ldots,J-1$. These row dependencies are independent: among them, only the $j$th uses rows from block row $j$. Hence $\rank H_P(E)\le JP-(J-1)$, and the rank--nullity formula $k=LP-\rank H_P(E)$ gives \eqref{eq:dimension}.

\subsection{A graph associated with a codeword}
\label{sec:codeword-graph}
We construct a graph whose vertices are exactly the 1-positions of a codeword. Each check contains an even number of 1s, so the vertices belonging to that check can be paired and joined by edges. Each edge thus represents two 1s canceling in a check equation. This construction applies to any exponent matrix $E$ and any of its codewords.

\begin{definition}[Graph associated with a codeword and pairings]
\label{def:codeword-graph}
Take a codeword $c\in C_P(E)$ with support $S$. At each check node $(j,s)$, pair the vertices of $S$ incident to that check, namely the positions $(\ell,t)\in S$ satisfying $t+e_{j\ell}\equiv s\pmod P$. Write $\Pi$ for the collection of all these pairing choices. The graph $G(c,\Pi)$ has vertex set $S$; each pair chosen at check $(j,s)$ is joined by an undirected edge of color $j$.
\end{definition}

Only variable nodes are vertices of $G(c,\Pi)$; check nodes are not retained as vertices. A pair contains two distinct vertices, so there are no loops. The same two vertices may be joined in different colors, so parallel edges are allowed. If a check contains four 1s, for example, there are three ways to pair them. Thus the graph depends on the pairing choices $\Pi$, as well as on the codeword. The following properties hold for every choice of pairings. The all-zero word gives the empty graph, with no vertices or edges.

\begin{figure}[p]
\centering
\input{figures/example34_symbols.tex}
\begin{minipage}{\textwidth}
\centering
\begin{tikzpicture}[x=1cm,y=1cm]
\node[inner sep=0,anchor=north west] at (0,0)
  {\includegraphics[width=16cm]{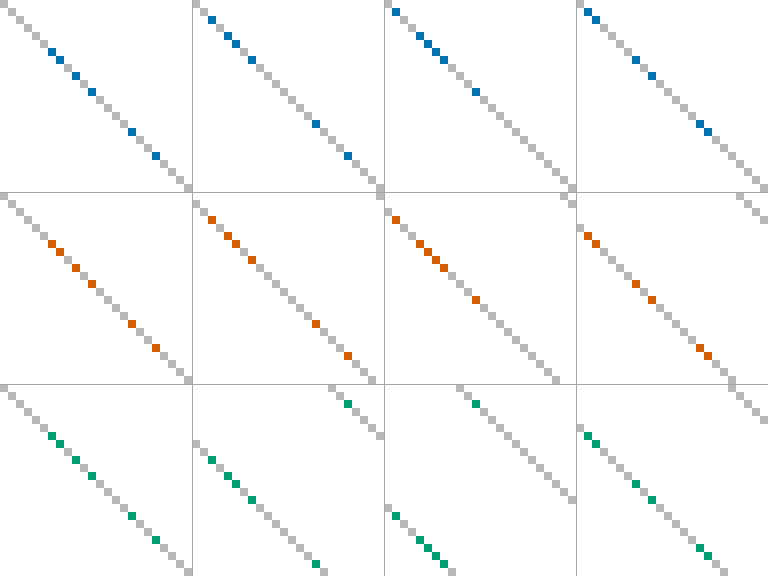}};
\foreach \cellcol in {1,...,95}
  \draw[black!15,line width=0.15pt] ({\cellcol/6},0) -- ({\cellcol/6},-12);
\foreach \cellrow in {1,...,71}
  \draw[black!15,line width=0.15pt] (0,{-\cellrow/6}) -- (16,{-\cellrow/6});
\draw[black!45,thin] (0,0) rectangle (16,-12);
\foreach \blockcol in {0,1,2,3}
  \node[font=\small] at ({(\blockcol+0.5)*4},0.3) {$\ell=\blockcol$};
\foreach \j in {0,1,2}
  \node[font=\small,anchor=east] at (-0.12,{-(\j+0.5)*4}) {$j=\j$};
\foreach \cellcol in {1,...,95}
  \draw[black!15,line width=0.15pt] ({\cellcol/6},-12.35) -- ({\cellcol/6},-12.85);
\draw[black!45,thin] (0,-12.35) rectangle (16,-12.85);
\draw[black!45,thin] (4,-12.35) -- (4,-12.85);
\draw[black!45,thin] (8,-12.35) -- (8,-12.85);
\draw[black!45,thin] (12,-12.35) -- (12,-12.85);
\begin{scope}[shift={(1.08333333,-12.6)},x=1.5mm,y=1.5mm]
\path[fill=exampleV0] (0,0) circle[radius=0.5];
\end{scope}
\begin{scope}[shift={(1.25000000,-12.6)},x=1.5mm,y=1.5mm]
\path[fill=exampleV0] (-0.5,-0.5) rectangle (0.5,0.5);
\end{scope}
\begin{scope}[shift={(1.58333333,-12.6)},x=1.5mm,y=1.5mm]
\path[fill=exampleV0] (0,0.5) -- (-0.5,-0.5) -- (0.5,-0.5) -- cycle;
\end{scope}
\begin{scope}[shift={(1.91666667,-12.6)},x=1.5mm,y=1.5mm]
\path[fill=exampleV0] (0,0.5) -- (-0.5,0) -- (0,-0.5) -- (0.5,0) -- cycle;
\end{scope}
\begin{scope}[shift={(2.75000000,-12.6)},x=1.5mm,y=1.5mm]
\path[fill=exampleV0] (0,-0.5) -- (-0.5,0.5) -- (0.5,0.5) -- cycle;
\end{scope}
\begin{scope}[shift={(3.25000000,-12.6)},x=1.5mm,y=1.5mm]
\path[fill=exampleV0] (-0.18,0.5) -- (0.18,0.5) -- (0.18,0.18) -- (0.5,0.18) -- (0.5,-0.18) -- (0.18,-0.18) -- (0.18,-0.5) -- (-0.18,-0.5) -- (-0.18,-0.18) -- (-0.5,-0.18) -- (-0.5,0.18) -- (-0.18,0.18) -- cycle;
\end{scope}
\begin{scope}[shift={(4.41666667,-12.6)},x=1.5mm,y=1.5mm]
\path[fill=exampleV1] (0,0) circle[radius=0.5];
\end{scope}
\begin{scope}[shift={(4.75000000,-12.6)},x=1.5mm,y=1.5mm]
\path[fill=exampleV1] (-0.5,-0.5) rectangle (0.5,0.5);
\end{scope}
\begin{scope}[shift={(4.91666667,-12.6)},x=1.5mm,y=1.5mm]
\path[fill=exampleV1] (0,0.5) -- (-0.5,-0.5) -- (0.5,-0.5) -- cycle;
\end{scope}
\begin{scope}[shift={(5.25000000,-12.6)},x=1.5mm,y=1.5mm]
\path[fill=exampleV1] (0,0.5) -- (-0.5,0) -- (0,-0.5) -- (0.5,0) -- cycle;
\end{scope}
\begin{scope}[shift={(6.58333333,-12.6)},x=1.5mm,y=1.5mm]
\path[fill=exampleV1] (0,-0.5) -- (-0.5,0.5) -- (0.5,0.5) -- cycle;
\end{scope}
\begin{scope}[shift={(7.25000000,-12.6)},x=1.5mm,y=1.5mm]
\path[fill=exampleV1] (-0.18,0.5) -- (0.18,0.5) -- (0.18,0.18) -- (0.5,0.18) -- (0.5,-0.18) -- (0.18,-0.18) -- (0.18,-0.5) -- (-0.18,-0.5) -- (-0.18,-0.18) -- (-0.5,-0.18) -- (-0.5,0.18) -- (-0.18,0.18) -- cycle;
\end{scope}
\begin{scope}[shift={(8.25000000,-12.6)},x=1.5mm,y=1.5mm]
\path[fill=exampleV2] (0,0) circle[radius=0.5];
\end{scope}
\begin{scope}[shift={(8.75000000,-12.6)},x=1.5mm,y=1.5mm]
\path[fill=exampleV2] (-0.5,-0.5) rectangle (0.5,0.5);
\end{scope}
\begin{scope}[shift={(8.91666667,-12.6)},x=1.5mm,y=1.5mm]
\path[fill=exampleV2] (0,0.5) -- (-0.5,-0.5) -- (0.5,-0.5) -- cycle;
\end{scope}
\begin{scope}[shift={(9.08333333,-12.6)},x=1.5mm,y=1.5mm]
\path[fill=exampleV2] (0,0.5) -- (-0.5,0) -- (0,-0.5) -- (0.5,0) -- cycle;
\end{scope}
\begin{scope}[shift={(9.25000000,-12.6)},x=1.5mm,y=1.5mm]
\path[fill=exampleV2] (0,-0.5) -- (-0.5,0.5) -- (0.5,0.5) -- cycle;
\end{scope}
\begin{scope}[shift={(9.91666667,-12.6)},x=1.5mm,y=1.5mm]
\path[fill=exampleV2] (-0.18,0.5) -- (0.18,0.5) -- (0.18,0.18) -- (0.5,0.18) -- (0.5,-0.18) -- (0.18,-0.18) -- (0.18,-0.5) -- (-0.18,-0.5) -- (-0.18,-0.18) -- (-0.5,-0.18) -- (-0.5,0.18) -- (-0.18,0.18) -- cycle;
\end{scope}
\begin{scope}[shift={(12.25000000,-12.6)},x=1.5mm,y=1.5mm]
\path[fill=exampleV3] (0,0) circle[radius=0.5];
\end{scope}
\begin{scope}[shift={(12.41666667,-12.6)},x=1.5mm,y=1.5mm]
\path[fill=exampleV3] (-0.5,-0.5) rectangle (0.5,0.5);
\end{scope}
\begin{scope}[shift={(13.25000000,-12.6)},x=1.5mm,y=1.5mm]
\path[fill=exampleV3] (0,0.5) -- (-0.5,-0.5) -- (0.5,-0.5) -- cycle;
\end{scope}
\begin{scope}[shift={(13.58333333,-12.6)},x=1.5mm,y=1.5mm]
\path[fill=exampleV3] (0,0.5) -- (-0.5,0) -- (0,-0.5) -- (0.5,0) -- cycle;
\end{scope}
\begin{scope}[shift={(14.58333333,-12.6)},x=1.5mm,y=1.5mm]
\path[fill=exampleV3] (0,-0.5) -- (-0.5,0.5) -- (0.5,0.5) -- cycle;
\end{scope}
\begin{scope}[shift={(14.75000000,-12.6)},x=1.5mm,y=1.5mm]
\path[fill=exampleV3] (-0.18,0.5) -- (0.18,0.5) -- (0.18,0.18) -- (0.5,0.18) -- (0.5,-0.18) -- (0.18,-0.18) -- (0.18,-0.5) -- (-0.18,-0.5) -- (-0.18,-0.18) -- (-0.5,-0.18) -- (-0.5,0.18) -- (-0.18,0.18) -- cycle;
\end{scope}
\node[font=\small,anchor=east] at (-0.12,-12.6) {$c^{\mathsf T}$};
\end{tikzpicture}
\par\smallskip
(a) Check matrix and codeword
\end{minipage}\par\medskip
\begin{minipage}{\textwidth}
\centering
\input{figures/example34_symbols.tex}
\definecolor{exampleJzero}{RGB}{0,114,178}
\definecolor{exampleJone}{RGB}{213,94,0}
\definecolor{exampleJtwo}{RGB}{0,158,115}
\begin{tikzpicture}[x=0.82cm,y=0.82cm,
  g0/.style={exampleJzero,line width=0.7pt},
  g1/.style={exampleJone,dashed,line width=0.85pt},
  g2/.style={exampleJtwo,densely dotted,line width=1pt,preaction={draw=white,solid,line width=2.2pt}},
  vertexlabel/.style={font=\scriptsize,inner sep=0pt,text=black}]
\coordinate (v0) at (-2.97861,2.56500);
\coordinate (v1) at (-1.02139,1.43500);
\coordinate (v2) at (2.97861,1.43500);
\coordinate (v3) at (1.02139,2.56500);
\coordinate (v4) at (2.00000,-3.13000);
\coordinate (v5) at (2.00000,-0.87000);
\coordinate (v6) at (-1.02139,-2.56500);
\coordinate (v7) at (-2.97861,-1.43500);
\coordinate (v8) at (-2.00000,3.13000);
\coordinate (v9) at (-2.00000,0.87000);
\coordinate (v10) at (2.97861,-2.56500);
\coordinate (v11) at (1.02139,-1.43500);
\coordinate (v12) at (-2.00000,-3.13000);
\coordinate (v13) at (-2.00000,-0.87000);
\coordinate (v14) at (-1.02139,2.56500);
\coordinate (v15) at (-2.97861,1.43500);
\coordinate (v16) at (2.97861,2.56500);
\coordinate (v17) at (1.02139,1.43500);
\coordinate (v18) at (-2.97861,-2.56500);
\coordinate (v19) at (-1.02139,-1.43500);
\coordinate (v20) at (2.00000,3.13000);
\coordinate (v21) at (2.00000,0.87000);
\coordinate (v22) at (2.97861,-1.43500);
\coordinate (v23) at (1.02139,-2.56500);
\draw[g0] (v12) -- (v18);
\draw[g0] (v6) -- (v19);
\draw[g0] (v7) -- (v13);
\draw[g0] (v8) -- (v14);
\draw[g0] (v0) -- (v15);
\draw[g0] (v1) -- (v9);
\draw[g0] (v16) -- (v20);
\draw[g0] (v2) -- (v21);
\draw[g0] (v3) -- (v17);
\draw[g0] (v10) -- (v22);
\draw[g0] (v4) -- (v23);
\draw[g0] (v5) -- (v11);
\draw[g1] (v6) -- (v12);
\draw[g1] (v7) -- (v18);
\draw[g1] (v0) -- (v8);
\draw[g1] (v13) -- (v19);
\draw[g1] (v1) -- (v14);
\draw[g1] (v9) -- (v15);
\draw[g1] (v2) -- (v16);
\draw[g1] (v3) -- (v20);
\draw[g1] (v17) -- (v21);
\draw[g1] (v4) -- (v10);
\draw[g1] (v5) -- (v22);
\draw[g1] (v11) -- (v23);
\draw[g2] (v11) -- (v17);
\draw[g2] (v0) .. controls (-4.17861,0.85500) and (-4.17861,-0.85500) .. (v18);
\draw[g2] (v1) -- (v19);
\draw[g2] (v2) .. controls (2.21096,-0.46402) and (0.02910,-0.94882) .. (v6);
\draw[g2] (v3) .. controls (-0.02910,0.94882) and (-2.21096,0.46402) .. (v7);
\draw[g2] (v8) -- (v20);
\draw[g2] (v9) -- (v21);
\draw[g2] (v4) -- (v12);
\draw[g2] (v5) -- (v13);
\draw[g2] (v14) .. controls (0.02910,0.94882) and (2.21096,0.46402) .. (v22);
\draw[g2] (v15) .. controls (-2.21096,-0.46402) and (-0.02910,-0.94882) .. (v23);
\draw[g2] (v10) .. controls (4.17861,-0.85500) and (4.17861,0.85500) .. (v16);
\fill[white] (v0) circle[radius=1.1mm];
\begin{scope}[shift={(v0)},x=2mm,y=2mm]
\path[fill=exampleV0] (0,0) circle[radius=0.5];
\end{scope}
\node[vertexlabel] at (-2.63220,2.36500) {$(0,6)$};
\fill[white] (v1) circle[radius=1.1mm];
\begin{scope}[shift={(v1)},x=2mm,y=2mm]
\path[fill=exampleV0] (-0.5,-0.5) rectangle (0.5,0.5);
\end{scope}
\node[vertexlabel] at (-1.36780,1.63500) {$(0,7)$};
\fill[white] (v2) circle[radius=1.1mm];
\begin{scope}[shift={(v2)},x=2mm,y=2mm]
\path[fill=exampleV0] (0,0.5) -- (-0.5,-0.5) -- (0.5,-0.5) -- cycle;
\end{scope}
\node[vertexlabel] at (2.63220,1.63500) {$(0,9)$};
\fill[white] (v3) circle[radius=1.1mm];
\begin{scope}[shift={(v3)},x=2mm,y=2mm]
\path[fill=exampleV0] (0,0.5) -- (-0.5,0) -- (0,-0.5) -- (0.5,0) -- cycle;
\end{scope}
\node[vertexlabel] at (1.36780,2.36500) {$(0,11)$};
\fill[white] (v4) circle[radius=1.1mm];
\begin{scope}[shift={(v4)},x=2mm,y=2mm]
\path[fill=exampleV0] (0,-0.5) -- (-0.5,0.5) -- (0.5,0.5) -- cycle;
\end{scope}
\node[vertexlabel] at (2.00000,-2.73000) {$(0,16)$};
\fill[white] (v5) circle[radius=1.1mm];
\begin{scope}[shift={(v5)},x=2mm,y=2mm]
\path[fill=exampleV0] (-0.18,0.5) -- (0.18,0.5) -- (0.18,0.18) -- (0.5,0.18) -- (0.5,-0.18) -- (0.18,-0.18) -- (0.18,-0.5) -- (-0.18,-0.5) -- (-0.18,-0.18) -- (-0.5,-0.18) -- (-0.5,0.18) -- (-0.18,0.18) -- cycle;
\end{scope}
\node[vertexlabel] at (2.00000,-1.27000) {$(0,19)$};
\fill[white] (v6) circle[radius=1.1mm];
\begin{scope}[shift={(v6)},x=2mm,y=2mm]
\path[fill=exampleV1] (0,0) circle[radius=0.5];
\end{scope}
\node[vertexlabel] at (-1.36780,-2.36500) {$(1,2)$};
\fill[white] (v7) circle[radius=1.1mm];
\begin{scope}[shift={(v7)},x=2mm,y=2mm]
\path[fill=exampleV1] (-0.5,-0.5) rectangle (0.5,0.5);
\end{scope}
\node[vertexlabel] at (-2.63220,-1.63500) {$(1,4)$};
\fill[white] (v8) circle[radius=1.1mm];
\begin{scope}[shift={(v8)},x=2mm,y=2mm]
\path[fill=exampleV1] (0,0.5) -- (-0.5,-0.5) -- (0.5,-0.5) -- cycle;
\end{scope}
\node[vertexlabel] at (-2.00000,2.73000) {$(1,5)$};
\fill[white] (v9) circle[radius=1.1mm];
\begin{scope}[shift={(v9)},x=2mm,y=2mm]
\path[fill=exampleV1] (0,0.5) -- (-0.5,0) -- (0,-0.5) -- (0.5,0) -- cycle;
\end{scope}
\node[vertexlabel] at (-2.00000,1.27000) {$(1,7)$};
\fill[white] (v10) circle[radius=1.1mm];
\begin{scope}[shift={(v10)},x=2mm,y=2mm]
\path[fill=exampleV1] (0,-0.5) -- (-0.5,0.5) -- (0.5,0.5) -- cycle;
\end{scope}
\node[vertexlabel] at (2.63220,-2.36500) {$(1,15)$};
\fill[white] (v11) circle[radius=1.1mm];
\begin{scope}[shift={(v11)},x=2mm,y=2mm]
\path[fill=exampleV1] (-0.18,0.5) -- (0.18,0.5) -- (0.18,0.18) -- (0.5,0.18) -- (0.5,-0.18) -- (0.18,-0.18) -- (0.18,-0.5) -- (-0.18,-0.5) -- (-0.18,-0.18) -- (-0.5,-0.18) -- (-0.5,0.18) -- (-0.18,0.18) -- cycle;
\end{scope}
\node[vertexlabel] at (1.36780,-1.63500) {$(1,19)$};
\fill[white] (v12) circle[radius=1.1mm];
\begin{scope}[shift={(v12)},x=2mm,y=2mm]
\path[fill=exampleV2] (0,0) circle[radius=0.5];
\end{scope}
\node[vertexlabel] at (-2.00000,-2.73000) {$(2,1)$};
\fill[white] (v13) circle[radius=1.1mm];
\begin{scope}[shift={(v13)},x=2mm,y=2mm]
\path[fill=exampleV2] (-0.5,-0.5) rectangle (0.5,0.5);
\end{scope}
\node[vertexlabel] at (-2.00000,-1.27000) {$(2,4)$};
\fill[white] (v14) circle[radius=1.1mm];
\begin{scope}[shift={(v14)},x=2mm,y=2mm]
\path[fill=exampleV2] (0,0.5) -- (-0.5,-0.5) -- (0.5,-0.5) -- cycle;
\end{scope}
\node[vertexlabel] at (-1.36780,2.36500) {$(2,5)$};
\fill[white] (v15) circle[radius=1.1mm];
\begin{scope}[shift={(v15)},x=2mm,y=2mm]
\path[fill=exampleV2] (0,0.5) -- (-0.5,0) -- (0,-0.5) -- (0.5,0) -- cycle;
\end{scope}
\node[vertexlabel] at (-2.63220,1.63500) {$(2,6)$};
\fill[white] (v16) circle[radius=1.1mm];
\begin{scope}[shift={(v16)},x=2mm,y=2mm]
\path[fill=exampleV2] (0,-0.5) -- (-0.5,0.5) -- (0.5,0.5) -- cycle;
\end{scope}
\node[vertexlabel] at (2.63220,2.36500) {$(2,7)$};
\fill[white] (v17) circle[radius=1.1mm];
\begin{scope}[shift={(v17)},x=2mm,y=2mm]
\path[fill=exampleV2] (-0.18,0.5) -- (0.18,0.5) -- (0.18,0.18) -- (0.5,0.18) -- (0.5,-0.18) -- (0.18,-0.18) -- (0.18,-0.5) -- (-0.18,-0.5) -- (-0.18,-0.18) -- (-0.5,-0.18) -- (-0.5,0.18) -- (-0.18,0.18) -- cycle;
\end{scope}
\node[vertexlabel] at (1.36780,1.63500) {$(2,11)$};
\fill[white] (v18) circle[radius=1.1mm];
\begin{scope}[shift={(v18)},x=2mm,y=2mm]
\path[fill=exampleV3] (0,0) circle[radius=0.5];
\end{scope}
\node[vertexlabel] at (-2.63220,-2.36500) {$(3,1)$};
\fill[white] (v19) circle[radius=1.1mm];
\begin{scope}[shift={(v19)},x=2mm,y=2mm]
\path[fill=exampleV3] (-0.5,-0.5) rectangle (0.5,0.5);
\end{scope}
\node[vertexlabel] at (-1.36780,-1.63500) {$(3,2)$};
\fill[white] (v20) circle[radius=1.1mm];
\begin{scope}[shift={(v20)},x=2mm,y=2mm]
\path[fill=exampleV3] (0,0.5) -- (-0.5,-0.5) -- (0.5,-0.5) -- cycle;
\end{scope}
\node[vertexlabel] at (2.00000,2.73000) {$(3,7)$};
\fill[white] (v21) circle[radius=1.1mm];
\begin{scope}[shift={(v21)},x=2mm,y=2mm]
\path[fill=exampleV3] (0,0.5) -- (-0.5,0) -- (0,-0.5) -- (0.5,0) -- cycle;
\end{scope}
\node[vertexlabel] at (2.00000,1.27000) {$(3,9)$};
\fill[white] (v22) circle[radius=1.1mm];
\begin{scope}[shift={(v22)},x=2mm,y=2mm]
\path[fill=exampleV3] (0,-0.5) -- (-0.5,0.5) -- (0.5,0.5) -- cycle;
\end{scope}
\node[vertexlabel] at (2.63220,-1.63500) {$(3,15)$};
\fill[white] (v23) circle[radius=1.1mm];
\begin{scope}[shift={(v23)},x=2mm,y=2mm]
\path[fill=exampleV3] (-0.18,0.5) -- (0.18,0.5) -- (0.18,0.18) -- (0.5,0.18) -- (0.5,-0.18) -- (0.18,-0.18) -- (0.18,-0.5) -- (-0.18,-0.5) -- (-0.18,-0.18) -- (-0.5,-0.18) -- (-0.5,0.18) -- (-0.18,0.18) -- cycle;
\end{scope}
\node[vertexlabel] at (1.36780,-2.36500) {$(3,16)$};
\draw[g0] (-2.7,-4.2) -- (-2.1,-4.2) node[right,black,font=\small] {$j=0$};
\draw[g1] (-0.7,-4.2) -- (-0.09999999999999998,-4.2) node[right,black,font=\small] {$j=1$};
\draw[g2] (1.3,-4.2) -- (1.9,-4.2) node[right,black,font=\small] {$j=2$};
\end{tikzpicture}
\par\smallskip
(b) Graph associated with the codeword
\end{minipage}
\caption{$J=3,L=4,P=24$ code and a weight-24 codeword. (a) The $72\times96$ check matrix $H_{24}(E_4)$ is above $c^{\mathsf T}$ in the same column order. The 72 colored entries are the 1s in the 24 support columns of $c$, highlighted in the color of block row $j$ used for the edges in (b). Other 1s are light gray, and 0s are white. The 24 symbols in the lower strip denote the 1s of $c$ and match the shapes and colors of the corresponding vertices in (b); blank positions denote 0. Vertex symbols use four block-column colors and six shapes assigned in increasing order of $t$ within each block. Light grid lines separate individual entries; darker gray lines separate the $24\times24$ blocks. Rows $(j,s)$ and columns $(\ell,t)$ are in lexicographic order. (b) Vertex labels give the 1-positions $(\ell,t)$, with an edge joining each pair at a check. Solid, dashed, and dotted edges have colors $j=0,1,2$, respectively. Line crossings are not vertices.}
\label{fig:example34}
\end{figure}

For example, take the first code in Table~\ref{tab:examples}, with exponent matrix and lift size
\begin{equation}
 E_4=\begin{pmatrix}
 0&0&0&0\\
 0&1&2&4\\
 0&7&15&5
 \end{pmatrix},\qquad P=24.
 \label{eq:example}
\end{equation}
Figure~\ref{fig:example34}(a) shows the $72\times96$ check matrix. Its column weight is 3, its row weight is 4, and the kernel in \eqref{eq:code} has dimension $k=4\cdot24-70=26$, using the rank 70. Its minimum distance 24 is verified by the exhaustive enumeration described in Appendix~\ref{app:computation}. Consider the codeword whose 1-positions in each block column are
\begin{equation*}
\begin{aligned}
T_0&=\{6,7,9,11,16,19\},\\
T_1&=\{2,4,5,7,15,19\},\\
T_2&=\{1,4,5,6,7,11\},\\
T_3&=\{1,2,7,9,15,16\}.
\end{aligned}

\end{equation*}
Thus $c_{\ell,t}=1$ exactly when $t\in T_\ell$. Each set has six elements, so $\wt(c)=24$. Counting the 1s at each check gives 38 checks with zero, 32 with two, and two with four. All counts are even, verifying $H_{24}(E_4)c=0$, the codeword condition in \eqref{eq:code}.

To specify $\Pi$, order the positions $(\ell,t)$ lexicographically at each check and pair consecutive positions. At check $(j,s)=(0,7)$, pair $(0,7)$ with $(1,7)$ and $(2,7)$ with $(3,7)$. At check $(1,6)$, pair $(0,6)$ with $(1,5)$ and $(2,4)$ with $(3,2)$. The pairing is unique at every other nonempty check. The graph in Fig.~\ref{fig:example34}(b) is connected and 3-regular, with 24 vertices and 36 edges; the 12 edges of each color form a perfect matching. For example, the neighbors of $(0,6)$ in colors 0, 1, and 2 are $(2,6)$, $(1,5)$, and $(3,1)$, respectively. On the color-1 edge, $5-6=0-1=e_{10}-e_{11}$, illustrating the relation between positions and exponents along an edge.
\FloatBarrier

\Needspace{6\baselineskip}
\begin{lemma}[Properties of the graph of a codeword]
\label{lem:graph-properties}
Put $w=\wt(c)$. The graph $G(c,\Pi)$ has the following properties.
\begin{enumerate}[label=\arabic*)]
\item Each vertex is incident to exactly one edge of each color. Thus each color forms a perfect matching: disjoint pairs covering every vertex exactly once. Each vertex has degree $J$, each color has $w/2$ edges, and the total number of edges is $Jw/2$.
\item Write the positions of vertices $u,v$ as $(\ell_u,t_u)$ and $(\ell_v,t_v)$. Giving an edge of color $j$ the orientation $u\to v$ for calculation yields
\begin{equation}
 t_v-t_u\equiv e_{j\ell_u}-e_{j\ell_v}\pmod P.
 \label{eq:edge}
\end{equation}
The column indices at the two ends of an edge are distinct.
\item Traverse a cycle as $u_0,u_1,\ldots,u_m=u_0$, and let $j_k$ be the color of edge $u_k\to u_{k+1}$. Then
\begin{equation}
 \sum_{k=0}^{m-1}\bigl(e_{j_k\ell_{u_k}}-e_{j_k\ell_{u_{k+1}}}\bigr)
 \equiv0\pmod P.
 \label{eq:graph-cycle}
\end{equation}
Thus the total position change around a cycle is zero modulo $P$. This includes cycles of length two formed by parallel edges.
\item The vector with 1s exactly at the vertices of any connected component is a nonzero codeword. Consequently, if $c$ is a nonzero codeword of minimum weight, then $G(c,\Pi)$ is connected.
\end{enumerate}
\end{lemma}
\begin{IEEEproof}
Each variable node is incident to exactly one check in each block row $j$ and occurs in exactly one pair there, proving 1). The two ends of an edge of color $j$ meet the same check, so $t_u+e_{j\ell_u}\equiv t_v+e_{j\ell_v}\pmod P$. This gives \eqref{eq:edge} in 2). If $\ell_u=\ell_v$, then $t_u=t_v$, making the two vertices identical; hence their column indices differ. Summing \eqref{eq:edge} around a cycle cancels each $t_u$ once with each sign, proving 3).

For 4), retain only one connected component and set all other positions to zero. The two vertices in each pair at a check are joined by an edge, so they are either both retained or both removed. Every check therefore still contains an even number of 1s and is satisfied. The component is nonempty, so the resulting codeword is nonzero. If the graph of a minimum-weight nonzero codeword were disconnected, retaining one component would give a nonzero codeword of smaller weight, a contradiction.
\end{IEEEproof}

The graph of a general codeword may be disconnected. In the proof of Theorem~\ref{thm:main}, choosing a nonzero codeword of minimum weight ensures connectivity by part 4) of Lemma~\ref{lem:graph-properties}. This allows a spanning tree in which all vertices can be reached from one root.

\section{Explicit construction}
\label{sec:explicit}
To attain the distance bound for every fixed $2\le J<L$, we use the bound on the number of exponents appearing in a cycle of a low-weight codeword. Choosing successive powers of a fixed integer controls cancellation in sums of exponents with small integer coefficients. Taking $P$ sufficiently large then turns the relevant congruences into integer equalities and allows us to exclude low-weight codewords. The following theorem gives the exponent matrix and a sufficient condition on $P$.

\Needspace{8\baselineskip}
\begin{theorem}[Attainment of the minimum-distance bound]
\label{thm:main}
Fix $2\le J<L$ and put $U=(J+1)!$. Define the fixed integer exponent matrix
\begin{equation}
 e_{j\ell}=U^{jL+\ell},\qquad
 R=U^{(J-1)L}(U^{L-1}-1).
 \label{eq:construction}
\end{equation}
Then, for every integer $P>(U-2)R$,
\begin{equation}
 \dmin(C_P(E))=U.
 \label{eq:main}
\end{equation}
\end{theorem}
The quantity $R$ in \eqref{eq:construction} is the largest difference between two entries in the same row. The array is independent of $P$; its entries are reduced modulo $P$ when forming the lift. This construction proves existence, but uses large exponents. The theorem also gives a lower bound $\dmin\ge D$ for any target $D\le(J+1)!$. For example, $J=2,3,4$ gives distances 6, 24, and 120, respectively, for every fixed $L>J$.

\begin{IEEEproof}[Proof of Theorem~\ref{thm:main}]
With $U=(J+1)!$ as in Theorem~\ref{thm:main}, the known upper bound is $\dmin\le U$~\cite[Corollary~9]{sv2011}, so it suffices to prove $\dmin\ge U$. Suppose a nonzero codeword of weight less than $U$ exists. We will associate a monomial with each of its 1-positions and obtain a polynomial vector of the same weight, contradicting the term-count bound in Lemma~\ref{lem:generic} of the appendix. Steps 1--3 establish the cycle property needed to preserve all check equations under this correspondence.

\emph{1) Choose the connected graph of a minimum-weight codeword.} Suppose $P>(U-2)R$ and a nonzero codeword of weight less than $U$ exists. Choose a nonzero codeword $c$ of minimum weight and put $w=\wt(c)$. The even-weight property proved in Section~\ref{sec:code-definition} gives $w\le U-2$ from $w<U$. Choose any pairings $\Pi$ as in Definition~\ref{def:codeword-graph} and form $G(c,\Pi)$. Part 4) of Lemma~\ref{lem:graph-properties} gives connectivity, and part 1) gives a perfect matching in each color.

\emph{2) Add the position changes around a cycle.} Choose a spanning tree, a set of edges joining all vertices without any cycle. There is a unique simple path in the tree between any two vertices. Adding an edge outside the tree therefore gives one cycle, consisting of that edge and the tree path. Write $m$ for the number of edges in this fundamental cycle; then $m\le w$. Each exponent difference along an edge uses two entries of the same row, so its absolute value is at most the maximum row range $R$ in \eqref{eq:construction}. The integer sum $s$ of the right-hand side of \eqref{eq:edge} around the cycle therefore satisfies $|s|\le wR\le(U-2)R<P$. The cycle condition \eqref{eq:graph-cycle} also gives $s\equiv0\pmod P$. Since the only multiple of $P$ with absolute value less than $P$ is zero, $s=0$ as an integer.

\emph{3) Show that each exponent has coefficient zero.} Substituting $e_{j\ell}=U^{jL+\ell}$ from \eqref{eq:construction} expresses this equality as
\begin{equation}
 \sum_{a=0}^{JL-1}q_aU^a=0,\qquad |q_a|\le w\le U-2.
 \label{eq:digits}
\end{equation}
In \eqref{eq:digits}, put $a=jL+\ell$; the integer $q_a$ counts the positive occurrences of $e_{j\ell}$ in the cycle sum minus its negative occurrences. Powers of $U$ are used so that the largest power cannot be canceled by all smaller powers combined. Indeed, if $h$ were the largest index with $q_h\ne0$, then
\begin{equation}
 \left|\sum_{a<h}q_aU^a\right|
 \le (U-2)\frac{U^h-1}{U-1}<U^h\le |q_h|U^h,
 \label{eq:dominance}
\end{equation}
This uses the coefficient bound in \eqref{eq:digits}, the geometric sum, and $|q_h|\ge1$ for the nonzero integer $q_h$. Inequality~\eqref{eq:dominance} contradicts $q_hU^h=-\sum_{a<h}q_aU^a$ in \eqref{eq:digits}. Thus all coefficients $q_a$ are zero. In other words, after collecting the exponent differences around each fundamental cycle, the integer coefficient of every exponent is zero. We use this equality of coefficients below.

The bound on the coefficients is essential here. For example, the chosen exponents satisfy $Ue_{00}-e_{01}=0$, but the coefficient $U$ is outside the permitted range $|q_a|\le U-2$. Every coefficient arising from a fundamental cycle of the low-weight codeword lies within that range, so the preceding estimate excludes any relation with a nonzero coefficient.

\emph{4) Represent the codeword positions and checks by polynomials.} Using step 3, we represent the 1-positions in each block column by polynomial terms. We first express the position change from the root to each vertex as an integer linear combination of the exponents, then associate a monomial with its coefficient vector.

\emph{Coefficient vectors along the spanning tree.} A common cyclic shift preserves the number of 1s at each corresponding check and the codeword weight. Subtract the root position from every position, so that the root has lift index zero. Continue to denote the shifted codeword and positions by $c$ and $t_u$. Sum \eqref{eq:edge} along the tree path from the root to $u$, and let $r_{u,j\ell}$ be the integer coefficient of $e_{j\ell}$ in this sum. These coefficients define a vector $r_u\in\Z^{JL}$. At the root the path is empty, giving $r_u=0$; at every vertex,
\begin{equation}
 t_u\equiv\sum_{j,\ell}r_{u,j\ell}e_{j\ell}\pmod P.
 \label{eq:position-map}
\end{equation}
Write $b_{j\ell}$ for the standard basis vector with its only nonzero entry at $(j,\ell)$. Traversing an edge $u\to v$ of color $j$ gives the coefficient vector $b_{j\ell_u}-b_{j\ell_v}$. Traversing the tree path from $u$ to $v$ gives $r_v-r_u$: reverse the path from the root to $u$, then follow the path from the root to $v$. The contributions from edges shared by the two paths cancel with opposite signs. For a tree edge, this already gives $r_v-r_u=b_{j\ell_u}-b_{j\ell_v}$.

For an edge $u\to v$ outside the tree, traverse that edge and then return from $v$ to $u$ along the tree. The coefficient vector of this fundamental cycle is
\begin{equation*}
 (b_{j\ell_u}-b_{j\ell_v})+(r_u-r_v)=0.
\end{equation*}
The equality to zero follows from step 3, where all coefficients $q_a$ in \eqref{eq:digits} were shown to vanish. Thus the coefficient vectors obtained by traversing an edge directly and by following the tree path agree, for edges both inside and outside the tree. Rearranging gives, for every edge of color $j$,
\begin{equation}
 r_u+b_{j\ell_u}=r_v+b_{j\ell_v}\quad\text{in }\Z^{JL}.
 \label{eq:coefficient-edge}
\end{equation}
The two sides are the coefficient vectors for the position changes from the root through either paired vertex to their common check. They agree as integer vectors, not merely after evaluating the positions modulo $P$. Consequently, the paired terms will remain equal when we associate an independent indeterminate with each exponent below.

\emph{From coefficient vectors to monomials.} Assign an independent indeterminate $z_{j\ell}$ to each exponent and write $Z=(z_{j\ell})$. Associate an integer vector $r=(r_{j\ell})$ with the monomial
\begin{equation}
 z^r=\prod_{j,\ell}z_{j\ell}^{r_{j\ell}}.
 \label{eq:monomial}
\end{equation}
Here $r_{j\ell}$ is an integer exponent, whereas the coefficient of the monomial is $1\in\F_2$. Coefficients are added modulo two, while exponents remain integers. For example, $z_{j\ell}+z_{j\ell}=0$, but $z_{j\ell}z_{j\ell}=z_{j\ell}^{2}$. Exponents are not reduced modulo $P$ either. The zero coefficient vector at the root corresponds to the monomial 1, and the difference $e_{j0}-e_{j1}$ corresponds to $z_{j0}z_{j1}^{-1}$. Polynomials allowing negative integer powers are called Laurent polynomials. Since the indeterminates are independent, $z^r=z^s$ holds exactly when $r=s$. Moreover, $z^{r+s}=z^rz^s$, so addition of coefficient vectors corresponds to multiplication of monomials.

For each block column $\ell$, add the monomials corresponding to its 1-positions, each with coefficient 1, to define
\begin{equation}
 p_\ell(Z)=\sum_{u:\ell_u=\ell}z^{r_u},\qquad
 p(Z)=(p_0(Z),\ldots,p_{L-1}(Z))^{\mathsf T}.
 \label{eq:polynomial-vector}
\end{equation}
Thus $p_\ell$ describes the occupied positions in that block column by its terms. Coefficients are added in the binary field, so two identical monomials cancel.

\emph{From parity checks to polynomial identities.} Moving from a variable in block column $\ell$ to its check in block row $j$ adds $e_{j\ell}$ to the position. This adds $b_{j\ell}$ to the coefficient vector and therefore multiplies its monomial by $z^{b_{j\ell}}=z_{j\ell}$. The polynomial sum of all contributions to checks in block row $j$ is thus $\sum_\ell z_{j\ell}p_\ell$. For vertices $u,v$ paired in color $j$, the coefficient equality \eqref{eq:coefficient-edge} and the definition of $z^r$ in \eqref{eq:monomial} give
\begin{equation*}
 z_{j\ell_u}z^{r_u}=z^{b_{j\ell_u}+r_u}=z^{b_{j\ell_v}+r_v}=z_{j\ell_v}z^{r_v}.
\end{equation*}
The two terms therefore cancel. For example, if the root is in column 0 and is paired in color $j$ with a vertex in column 1, their contributions to the check are
\begin{equation*}
 z_{j0}\cdot1+z_{j1}\cdot(z_{j0}z_{j1}^{-1})
 =z_{j0}+z_{j0}=0.
\end{equation*}
This is the polynomial expression of two 1s contributing even parity at the same check.

Substituting the definition of $p_\ell$ in \eqref{eq:polynomial-vector} and counting each vertex once in its block column gives
\begin{equation}
 \sum_\ell z_{j\ell}p_\ell
 =\sum_\ell\sum_{u:\ell_u=\ell}z_{j\ell}z^{r_u}
 =\sum_u z^{b_{j\ell_u}+r_u}=0.
 \label{eq:polynomial-checks}
\end{equation}
In the final sum, the two terms in each color-$j$ pair are equal. Every vertex belongs to exactly one such pair, so all terms cancel. The matrix equation $Zp=0$ collects \eqref{eq:polynomial-checks} for all $j$.

\emph{From weight to term count.} Cancellation within a check equation must be distinguished from cancellation within an individual $p_\ell$. If distinct vertices in the same column had $r_u=r_v$, \eqref{eq:position-map} would give $t_u\equiv t_v\pmod P$. Since $t_u,t_v$ are elements of $\Z/P\Z$, they would represent the same position, making them the same variable node. Thus the terms within $p_\ell$ are distinct, and their number equals the number of 1s in that block column. Terms from different columns belong to different components and cannot cancel each other. Defining the weight of a polynomial vector as the total number of terms in its components, we have constructed
\begin{equation}
 Zp=0,\qquad p\ne0,\qquad \wt(p)=\wt(c)=w.
 \label{eq:weight-preservation}
\end{equation}
Lemma~\ref{lem:generic} in the appendix states that any nonzero vector satisfying these polynomial check equations has at least $U$ terms. The check equations \eqref{eq:polynomial-checks} and the nonzeroness and weight equality in \eqref{eq:weight-preservation} allow us to apply \eqref{eq:generic} to $p$, giving $w=\wt(p)\ge U$, contrary to the assumption $w<U$. Hence $\dmin\ge U$, and equality follows from the known upper bound \cite[Corollary~9]{sv2011}.

\end{IEEEproof}

\section{Random construction}
\label{sec:random}
The explicit construction guarantees the distance by choosing widely separated exponents. We now choose each exponent independently and uniformly from $0,\ldots,P-1$. A low-weight codeword still requires the cycle conditions used in the preceding section. By bounding the probability of each condition and the total number of conditions, we show that the probability of falling below the distance bound tends to zero as $P$ grows.

\begin{corollary}[Uniform random exponents]
\label{cor:random}
Fix $2\le J<L$ and choose the entries of $E$ independently and uniformly from $\Z/P\Z$. With $U=(J+1)!$, for every integer $P\ge2$ we have
\begin{equation}
 \Pr[\dmin(C_P(E))<U]\le\min\{1,K_{J,L}/P\},
 \label{eq:probability}
\end{equation}
Here $K_{J,L}$ depends only on $J,L$; an explicit value is given in \eqref{eq:constant}.
\end{corollary}
\begin{IEEEproof}
\emph{1) Fix cycle conditions before sampling exponents.} For each even $w$ with $2\le w<U$, fix $J$ perfect matchings on $w$ labeled vertices and a block-column index for each vertex. Retain only connected graphs, and fix a spanning tree, an ordering of its fundamental cycles, and a direction in which to traverse each cycle. For each fundamental cycle, form its integer coefficient vector by counting the positive occurrences of each exponent in \eqref{eq:graph-cycle} minus its negative occurrences. These counts depend only on the edge colors and the block-column indices of the vertices; they use neither exponent values nor vertex lift indices. If any coefficient vector is nonzero, let $q$ be the first such vector in the fixed ordering. Thus $q$ is determined before the exponents are sampled.

For a fixed choice, let $\mathcal A$ be the event that lift indices can be assigned to its vertices so that they occupy distinct variable-node positions in the prescribed block columns and the two vertices in each pair meet the same check. Such an assignment gives a weight-$w$ codeword, because the contributions to each check cancel pair by pair; we then say that the choice is realized. Suppose a choice with all coefficient vectors zero were realized. The construction in step 4 of the proof of Theorem~\ref{thm:main} would produce a nonzero polynomial vector $p$ satisfying \eqref{eq:weight-preservation}. Distinct vertices in the same block column have distinct lift indices, so no terms cancel within a component and the total number of terms is $w$. This contradicts Lemma~\ref{lem:generic}. Hence, if all coefficient vectors are zero, $\mathcal A$ is empty and its probability is zero.

For every other choice, the nonzero vector $q$ fixed in step 1 exists. Whichever assignment of lift indices realizes that choice, summing \eqref{eq:edge} around the selected cycle cancels the lift indices and yields the same congruence $q\cdot E\equiv0\pmod P$. Thus
\begin{equation}
 \mathcal A\subseteq\{E:q\cdot E\equiv0\pmod P\}.
 \label{eq:realization-condition}
\end{equation}
In \eqref{eq:realization-condition}, $q\cdot E=\sum_{j,\ell}q_{jL+\ell}e_{j\ell}$. The indexing $a=jL+\ell$ is the same as in \eqref{eq:digits}, but here $q$ is fixed before choosing numerical exponents. A fundamental cycle has at most $w$ edges, so $|q_a|\le w$.

\emph{2) Bound the probability of one cycle condition.} Let $g=\gcd(P,q_0,\ldots,q_{JL-1})$ be the greatest common divisor. Every remainder of $q\cdot E$ is a multiple of $g$. Conversely, B\'ezout's identity expresses $g$ as an integer linear combination of $P,q_0,\ldots,q_{JL-1}$, so $g$ and all its multiples are attained modulo $P$. The possible remainders are therefore $0,g,2g,\ldots,P-g$, giving $P/g$ values. For any two of these remainders, adding a suitable fixed matrix entrywise modulo $P$ gives a bijection between the exponent matrices producing them. Since $E$ is uniform over all exponent matrices, the remainders are equally likely, and zero occurs with probability $g/P$. Choose a nonzero coefficient $q_a$. Since $g$ is a positive divisor of $|q_a|$, we have $1\le g\le|q_a|\le w$. Hence
\begin{equation}
 \Pr[q\cdot E\equiv0\pmod P]
 =\frac{\gcd(P,q_0,\ldots,q_{JL-1})}{P}\le\frac{w}{P}.
 \label{eq:gcd}
\end{equation}
\emph{3) Count graphs and column assignments.} There are $L^w$ assignments of block-column indices to the $w$ labeled vertices. To form a perfect matching of one color, choose a partner for the vertex with the smallest label in $w-1$ ways, remove the pair, and repeat. This gives $(w-1)!!=(w-1)(w-3)\cdots1$ perfect matchings, and $((w-1)!!)^J$ choices for $J$ colors. This count includes disconnected graphs and assignments that cannot be realized by codewords, so the total number of choices considered in step 1 is at most $L^w((w-1)!!)^J$.

The event inclusion \eqref{eq:realization-condition} and the probability bound \eqref{eq:gcd} give $\Pr[\mathcal A]\le w/P$ for each choice. The event $\mathcal A$ already allows the existence of any assignment of lift indices. We have bounded one necessary condition shared by all these assignments, so no factor $P^w$ is needed. If $\dmin(C_P(E))<U$, a minimum-weight nonzero codeword has even weight below $U$ and a connected graph by Lemma~\ref{lem:graph-properties}. Labeling its vertices realizes one of the choices counted in step 3. Thus the event that the distance is below $U$ is contained in the union of these realization events. Apply the union bound to sum the probability bounds over these choices. No independence between events is required. Define the part that does not depend on $P$ as
\begin{equation}
 K_{J,L}=\sum_{\substack{2\le w<U\\w\text{ even}}}wL^w((w-1)!!)^J.
 \label{eq:constant}
\end{equation}
The probability estimate just described becomes
\begin{equation}
\begin{split}
 \Pr[\dmin(C_P(E))<U]
 &\le\sum_{\substack{2\le w<U\\w\text{ even}}}
 L^w((w-1)!!)^J\frac{w}{P}\\
 &=\frac1P\sum_{\substack{2\le w<U\\w\text{ even}}}
 wL^w((w-1)!!)^J=\frac{K_{J,L}}{P}.
\end{split}
 \label{eq:union-bound}
\end{equation}
Here we used \eqref{eq:constant}. Together with the upper bound one on any probability, this gives \eqref{eq:probability}.
\end{IEEEproof}
Equation~\eqref{eq:gcd} does not require $P$ to be prime, so the probability bound applies to composite lift sizes as well. The constant $K_{J,L}$ is large because it counts all pairings and column assignments; \eqref{eq:probability} gives only the trivial upper bound one when $P\le K_{J,L}$. It therefore does not predict success rates at small $P$. Removing duplicate conditions can improve the bound. Examples at small lift sizes are given in Section~\ref{sec:examples}. Numerical sufficient lift sizes for both constructions are given in Appendix~\ref{app:lift-thresholds}.

\Needspace{12\baselineskip}
\section{Examples}
\label{sec:examples}
The explicit construction guarantees attainment of the distance bound, and the random construction gives the attainment probability at large $P$. At small lift sizes, we aim to attain the bound 24 for $J=3$ and exceed 24 for $J=4$. Following the general construction, we list arrays of distance 24 for $J=3$ and distance at least 28 for $J=4$ in a common table.
\begin{table}[htp]
\centering
\normalsize
\setlength{\tabcolsep}{6pt}
\setlength{\arraycolsep}{4pt}
\renewcommand{\arraystretch}{0.92}
\caption{Exponent matrices and parameters of binary CPM-LDPC codes. The $J=3$ examples attain the minimum-distance upper bound $(J+1)!=24$, whereas the $J=4$ examples do not attain the upper bound $(J+1)!=120$.}
\label{tab:examples}
\begin{tabular}{ccrrrcl}
\toprule
$J$ & $L$ & $P$ & $n$ & $k$ & $d$ & $E=(e_{j\ell})$\\
\midrule
3 & 4 & 24 & 96 & 26 & $24$ & $\begin{pmatrix}0 & 0 & 0 & 0\\0 & 1 & 2 & 4\\0 & 7 & 15 & 5\end{pmatrix}$\\
\addlinespace[4pt]
3 & 5 & 45 & 225 & 92 & $24$ & $\begin{pmatrix}0 & 0 & 0 & 0 & 0\\0 & 1 & 3 & 10 & 14\\0 & 40 & 31 & 33 & 30\end{pmatrix}$\\
\addlinespace[4pt]
3 & 6 & 71 & 426 & 215 & $24$ & $\begin{pmatrix}0 & 0 & 0 & 0 & 0 & 0\\0 & 47 & 24 & 69 & 39 & 55\\0 & 61 & 41 & 19 & 13 & 58\end{pmatrix}$\\
\addlinespace[4pt]
3 & 7 & 111 & 777 & 446 & $24$ & $\begin{pmatrix}0 & 0 & 0 & 0 & 0 & 0 & 0\\0 & 3 & 11 & 15 & 45 & 93 & 110\\0 & 34 & 18 & 9 & 1 & 4 & 101\end{pmatrix}$\\
\addlinespace[4pt]
3 & 8 & 159 & 1272 & 797 & $24$ & $\begin{pmatrix}0 & 0 & 0 & 0 & 0 & 0 & 0 & 0\\0 & 22 & 62 & 117 & 79 & 126 & 70 & 123\\0 & 11 & 67 & 13 & 129 & 87 & 102 & 109\end{pmatrix}$\\
\midrule
4 & 5 & 23 & 115 & 26 & $30$ & $\begin{pmatrix}0 & 0 & 0 & 0 & 0\\0 & 1 & 3 & 10 & 14\\0 & 17 & 8 & 11 & 6\\0 & 10 & 14 & 12 & 4\end{pmatrix}$\\
\addlinespace[4pt]
4 & 6 & 29 & 174 & 61 & $28$ & $\begin{pmatrix}0 & 0 & 0 & 0 & 0 & 0\\0 & 18 & 24 & 11 & 10 & 26\\0 & 3 & 12 & 19 & 13 & 28\\0 & 26 & 16 & 24 & 25 & 27\end{pmatrix}$\\
\addlinespace[4pt]
4 & 7 & 43 & 301 & 132 & $32$ & $\begin{pmatrix}0 & 0 & 0 & 0 & 0 & 0 & 0\\0 & 3 & 11 & 15 & 2 & 7 & 24\\0 & 34 & 18 & 9 & 1 & 4 & 15\\0 & 37 & 41 & 20 & 28 & 25 & 2\end{pmatrix}$\\
\addlinespace[4pt]
4 & 8 & 73 & 584 & 295 & $30\le d\le46$ & $\begin{pmatrix}0 & 0 & 0 & 0 & 0 & 0 & 0 & 0\\0 & 22 & 62 & 44 & 6 & 53 & 70 & 50\\0 & 11 & 67 & 13 & 56 & 14 & 29 & 36\\0 & 23 & 45 & 2 & 29 & 38 & 69 & 46\end{pmatrix}$\\
\bottomrule
\end{tabular}
\par\smallskip
\begin{minipage}{\textwidth}
\small
\textit{Code definition.} $E$ is the $J\times L$ exponent matrix and $P$ is the lift size. Let $S_P$ be the $P\times P$ circulant permutation matrix whose 1 in column $t$ is in row $(t+1\bmod P)$, for $0\le t<P$. The check matrix and binary code are
\[ H_P(E)=[S_P^{\,e_{j\ell}}]_{j,\ell},\qquad C_P(E)=\{c\in\F_2^{LP}:H_P(E)c=0\}. \]
Exponents are taken modulo $P$; exponent 0 denotes the identity matrix $S_P^0=I_P$. Here $J$ is the column weight, $L$ the row weight, $n=LP$ the length, $k=LP-\rank_{\F_2}H_P(E)$ the dimension, and $d$ the minimum distance. Numerical entries for $d$ are exact; the inequality gives verified lower and upper bounds.
\end{minipage}
\end{table}

For every integer $L\ge4$, Theorem~\ref{thm:main} with $J=3$ gives
\begin{equation}
 (E_L)_{j\ell}=24^{jL+\ell}\quad(0\le j<3,\ 0\le\ell<L),\qquad
 P_L=22\cdot24^{2L}(24^{L-1}-1)+1
 \label{eq:example-general}
\end{equation}
with minimum distance 24 for every $L\ge4$. The individual examples below use much smaller lift sizes, with their distances established by exact verification.

The first code in Table~\ref{tab:examples} is the $[96,26,24]$ code in \eqref{eq:example}, illustrated in Section~\ref{sec:codeword-graph}. It attains the bound at the composite lift size $P=24$, with its minimum distance verified by exhaustive enumeration.

The $J=3$ arrays for $L=5,7$ are reverified examples from Bocharova et al.~\cite[Table~III]{bocharova2011}. For $L=6,8$, searching at lift sizes below their values $P=72,160$ yielded the arrays at $P=71,159$ in the table.

For $J=4$, we started with candidates formed by adding a fourth row to known $J=3$ arrays, then decreased $P$ in stages and reselected the exponents. For every pair of distinct block rows $j,j'$, the differences $e_{j\ell}-e_{j'\ell}$ are distinct modulo $P$ across columns. This condition excludes 4-cycles~\cite{fossorier2004}. We selected candidates with no codeword of weight at most 24.

All $J=3$ examples have verified minimum distance 24. For $J=4$, the exact distances for $L=5,6,7$ are 30, 28, and 32, respectively, while the verified interval for $L=8$ is $30\le d\le46$. Each lower bound follows from exhaustive enumeration or a complete search at smaller weights, and each upper bound is witnessed by a codeword.

The reported lift sizes have not been proved minimal. Appendix~\ref{app:computation} explains search completeness and reproduction. The two examples for each $L=5,\ldots,8$ also differ in length and rate $k/n$. They illustrate attainment of the bound at column weight three and constructions exceeding that distance at column weight four.
\FloatBarrier

\section{Conclusion}
For binary CPM-LDPC codes with a full array of single circulant permutation matrices, the distance bound $(J+1)!$ is attainable for every fixed $2\le J<L$ and every sufficiently large integer lift size. The explicit construction guarantees this with one fixed exponent matrix, while the random construction shows that independent uniform choices attain the bound with probability tending to one. Both results use cycle conditions necessary for low-weight codewords, through deterministic exclusion and probability bounds, respectively.

The small-lift examples attain the bound 24 for $J=3$ and achieve distance at least 28 for $J=4$. The examples with $L=5,P=23$, $L=6,P=29$, and $L=7,P=43$ have exact distances 30, 28, and 32, respectively. A direction for further work is to reduce the exponent range and lift size while retaining this distance. For additional linear constraints on the exponents, a starting point is to determine whether those constraints make a cycle condition hold identically.

For the random construction, removing pairings that give the same condition and choices that cannot be realized may reduce the constant in \eqref{eq:constant}. Reducing the exponent range of the explicit construction and improving the probability bound are directions for bringing the general distance guarantees closer to the small-lift examples.

\appendices
\section{Term count for polynomial check equations}
\label{app:algebra}
Equation~\eqref{eq:weight-preservation} in the proof of Theorem~\ref{thm:main} constructs, from a weight-$w$ codeword, a nonzero polynomial vector $p$ satisfying $Zp=0$ with exactly $w$ terms across its components. It remains to show that such a vector requires at least $(J+1)!$ terms. Below, we write $c$ for an arbitrary polynomial vector. The theorem of Draisma, Kahle, and Wiersig bounds the term count of a single polynomial. We therefore combine the components of $c$ into one polynomial $f$ while preserving their total term count, and then use $Zc=0$ to establish the vanishing condition required by that theorem.

\setcounter{theorem}{0}
\renewcommand{\thetheorem}{A.\arabic{theorem}}
\begin{lemma}[Term count for polynomial check equations]
\label{lem:generic}
Let $2\le J<L$, and let $Z=(z_{j\ell})$ be a $J\times L$ matrix of independent indeterminates. Let $\mathcal R=\F_2[z_{j\ell}^{\pm1}]$ be the Laurent polynomial ring, which also allows negative integer powers. Define the weight of $c\in\mathcal R^L$ as the sum of the numbers of monomials in its components. The condition $Zc=0$ means that $\sum_\ell z_{j\ell}c_\ell=0$ in every row $j$: every term occurs an even number of times and cancels. Then
\begin{equation}
 \min_{0\ne c\in\ker_{\mathcal R}Z}\wt(c)=(J+1)!.
 \label{eq:generic}
\end{equation}
\end{lemma}
\begin{IEEEproof}
First prove the lower bound. Take $0\ne c\in\ker_{\mathcal R}Z$. Only finitely many exponents occur in its components, so multiplying by a common monomial $m$ makes every exponent nonnegative. This operation adds the same vector to the exponent vector of every monomial. Distinct terms therefore remain distinct, preserving both nonzeroness and the term count in each component. It also preserves the check equations, since $Z(mc)=m(Zc)=0$. Continue to denote this vector of ordinary polynomials by $c$.

To express the total term count of the components as the term count of one polynomial, we must keep terms from different components distinct. Simply adding the components could cancel a monomial that occurs in more than one component. Introduce new indeterminates $y_0,\ldots,y_{L-1}$, independent of the $z_{j\ell}$, and multiply each component by a different indeterminate:
\begin{equation}
 f(Z,y)=\sum_{\ell=0}^{L-1}y_\ell c_\ell(Z).
 \label{eq:augmented}
\end{equation}
For example, a monomial $m$ shared by $c_0,c_1$ gives the distinct terms $y_0m,y_1m$ in $f$. Distinct terms within a component also remain distinct after multiplication by $y_\ell$. Thus the terms in the components of $c$ correspond one to one to the terms of $f$. In particular, $f\ne0$ and its term count is $\wt(c)$.

Appending the row $y=(y_0,\ldots,y_{L-1})$ to $Z$ makes $f$ a polynomial in the entries of a $(J+1)\times L$ matrix. The Draisma--Kahle--Wiersig theorem \cite[Theorem~3.1]{dkw2023} states that, over an algebraically closed field, a nonzero polynomial vanishing on every $m\times n$ matrix of rank $r$ has at least $(r+1)!$ terms. The theorem also holds in characteristic two. To apply it with $m=J+1$, $n=L$, and $r=J$, we show that $f(Z,y)=0$ whenever the matrix obtained by appending $y$ to $Z$ has rank at most $J$.

If $y$ is a linear combination $y=aZ$ of the rows of $Z$, then \eqref{eq:augmented} and $Zc=0$ give $f=yc=a(Zc)=0$. This corresponds to a codeword satisfying every linear combination of its parity checks. If the rank of $Z$ is less than $J$, however, the appended matrix can have rank at most $J$ even when $y$ is outside the row space of $Z$. A matrix factorization and a univariate polynomial allow us to include this case as well.

Work over an algebraic closure $K$ of $\F_2$, an infinite field containing $\F_2$ in which every nonconstant univariate polynomial has a root. Viewing $f,c$ as polynomials over $K$ changes neither their coefficients nor their monomials, so $f$ remains nonzero with the same term count. Moreover, $Zc(Z)=0$ is a polynomial identity: the coefficient of every monomial is zero. It therefore remains valid under every assignment of values in $K$. For example, the nonzero polynomial $x^2+x$ evaluates to zero at both elements of $\F_2$. Thus checking only assignments in $\F_2$ would not establish a polynomial identity.

In verifying the vanishing condition, use the same symbols $Z,y$ for matrices and rows obtained by assigning values in $K$. Suppose the matrix obtained by appending $y$ to $Z$ has rank at most $J$. Place a basis of its column space in a left factor and the coefficients expressing each column in that basis in a right factor. The basis has at most $J$ elements. Padding the left factor with zero columns and the right factor with zero rows if necessary gives
\begin{equation}
 \begin{pmatrix}Z\\y\end{pmatrix}
 =\begin{pmatrix}A\\a\end{pmatrix}B,
 \qquad A\in K^{J\times J},\quad a\in K^{1\times J},\quad B\in K^{J\times L}.
 \label{eq:rank-factorization}
\end{equation}
Equation~\eqref{eq:rank-factorization} gives $Z=AB$ and $y=aB$. If $A$ is invertible, then $y=aA^{-1}Z$ is a linear combination of the rows of $Z$, so the preceding argument gives $f(Z,y)=0$.

If $A$ in \eqref{eq:rank-factorization} is singular, keep $B,a$ fixed and replace $A$ by $A+sI_J$. Substitute this matrix into $f$ from \eqref{eq:augmented} and define
\begin{equation}
 h(s)=f((A+sI_J)B,aB).
 \label{eq:auxiliary-polynomial}
\end{equation}
Here $I_J$ is the $J\times J$ identity matrix and $s$ is a new indeterminate. Because $f$ is an ordinary polynomial, $h$ has no denominators and is defined for every $s$, including $s=0$. The determinant $\det(A+sI_J)$ is monic of degree $J$ and therefore has only finitely many roots. For every other value of $s$, the matrix $A+sI_J$ is invertible, and the preceding argument applied to $Z=(A+sI_J)B$ and $y=aB$ gives $h(s)=0$. Since $K$ is infinite, $h$ has infinitely many roots. A nonzero univariate polynomial has only finitely many roots, so $h$ is identically zero. Evaluating \eqref{eq:auxiliary-polynomial} at $s=0$ and using $AB=Z$, $aB=y$ from \eqref{eq:rank-factorization} gives $0=h(0)=f(AB,aB)=f(Z,y)$. This proves the vanishing condition for singular $A$ as well.

Thus the nonzero polynomial $f$ vanishes on every $(J+1)\times L$ matrix of rank at most $J$. Applying the Draisma--Kahle--Wiersig theorem~\cite[Theorem~3.1]{dkw2023} to $f$ in \eqref{eq:augmented} gives at least $(J+1)!$ terms. Since its term count equals $\wt(c)$, this proves the lower bound $\wt(c)\ge(J+1)!$.

A cofactor vector attains equality in \eqref{eq:generic}. Choose $J+1$ columns of $Z$. At each selected column, let the corresponding component be the determinant of the $J\times J$ submatrix on the other $J$ selected columns. Set all unselected components to zero. In characteristic two, $-1=1$, so cofactor signs can be omitted.

To verify $Zc=0$, fix a row $j$. Take the matrix on the selected $J+1$ columns and append another copy of its $j$th row. The resulting square matrix has two identical rows, so its determinant is zero. Expanding this determinant along the appended row gives precisely $\sum_\ell z_{j\ell}c_\ell$. Hence the check equation holds for every row $j$.

The determinant in each selected component is a sum of $J!$ terms, one for each permutation of the $J$ columns. Distinct permutations choose different column indeterminates in at least one row. Since the indeterminates are independent, the resulting monomials are all distinct and none cancel. Thus all $J+1$ selected components are nonzero and each has $J!$ terms. Their total term count is $(J+1)J!=(J+1)!$, attaining the lower bound. For example, when $J=2$, three components have two terms each, giving six terms.
\end{IEEEproof}

\section{Computational distance verification}
\label{app:computation}
The distance lower bounds in Table~\ref{tab:examples} are verified by a search over nonzero codeword supports. A common cyclic shift preserves the code, so a support position in its first occupied block column $\ell$ can be moved to $(\ell,0)$. Forbidding earlier block columns and considering all $\ell=0,\ldots,L-1$ therefore includes a representative of a minimum-weight codeword. This symmetry also holds for composite $P$.

Choose a check containing an odd number of positions in the current partial support. A completed codeword must add at least one remaining position incident to it. Branch on those candidates in a fixed order, forbidding candidates from earlier branches in each later branch. Prune only when the checks cannot be satisfied within the remaining weight budget. The arrays, source, run records, and pruning arguments are available in the public repository~\cite{kasai2026code}.

For $J=3,L=4,P=24$, two independent implementations enumerated all $2^{26}-1$ nonzero codewords, confirming minimum distance 24 and exactly 668 words of weight 24. For $J=4,L=5,P=23$, an independent implementation also enumerated all $2^{26}-1$ nonzero codewords, confirming minimum distance 30 and exactly 23 words of weight 30.

For each $J=3$ array with $L=5,\ldots,8$, searches completed through weight 22 for every root and partition, and a weight-24 codeword was found. For $J=4,L=6,7,8$, searches were completed through weights 26, 30, and 28, respectively, and codewords of weights 28, 32, and 46 were found. By the even-weight property in Section~\ref{sec:code-definition}, the distances for $L=6,7$ are exactly 28 and 32, while $30\le d\le46$ for $L=8$.

A separate checker verifies stored run-record hashes and partition coverage, exponent arrays, parity-check ranks, and the codewords providing upper bounds. Checking these records is distinct from rerunning the search. Reproducing the lower bounds requires rerunning the search implementation in this repository.

\Needspace{8\baselineskip}
\section{Numerical sufficient lift sizes}
\label{app:lift-thresholds}
Table~\ref{tab:lift-thresholds} evaluates the sufficient conditions in Theorem~\ref{thm:main} and Corollary~\ref{cor:random} for $J=2,3,4$ and $J<L\le8$. With $U=(J+1)!$, $R$ from \eqref{eq:construction}, and $K_{J,L}$ from \eqref{eq:constant}, put
\begin{equation}
 P_{\mathrm{exp}}=(U-2)R+1,\qquad
 P_{+}=K_{J,L}+1.
 \label{eq:lift-thresholds}
\end{equation}
For every integer $P\ge P_{\mathrm{exp}}$, the construction with exponents $e_{j\ell}=U^{jL+\ell}$ has minimum distance $U$. For every integer $P\ge P_{+}$, choosing all exponent entries independently and uniformly from $0,\ldots,P-1$ gives minimum distance $U$ with probability at least $1-K_{J,L}/P>0$. The value $P_{+}$ in \eqref{eq:lift-thresholds} is the smallest integer $P$ for which this probability lower bound is positive. Neither condition requires $P$ to be prime.

To obtain the success probability, use the known upper bound $\dmin\le U$: the events $\dmin=U$ and $\dmin<U$ are complementary. Equation~\eqref{eq:probability} therefore gives
\begin{equation}
 \Pr[\dmin=U]
 =1-\Pr[\dmin<U]
 \ge1-\min\{1,K_{J,L}/P\}
 =\max\{0,1-K_{J,L}/P\}.
 \label{eq:success-probability}
\end{equation}
For integers $P,K_{J,L}$, positivity of the lower bound in \eqref{eq:success-probability} is equivalent to
\begin{equation}
 1-K_{J,L}/P>0
 \quad\Longleftrightarrow\quad P-K_{J,L}>0
 \quad\Longleftrightarrow\quad P\ge K_{J,L}+1.
 \label{eq:positive-threshold}
\end{equation}
This gives $P_{+}$ in \eqref{eq:lift-thresholds}.

For example, when $J=2$, the sum in \eqref{eq:constant} contains only $w=2,4$, giving
\begin{equation*}
 K_{2,L}=2L^2+36L^4.
\end{equation*}
For $L=3$, this gives $K_{2,3}=2934$ and $P_{+}=2935$, with guaranteed success probability at least $1/2935$.

The table rounds each threshold upward to three significant digits, preserving the sufficient conditions. These values are neither minimum feasible lift sizes nor measured success rates. In particular, $K_{J,L}$ is large because its count includes repeated conditions and unrealizable choices. The $J=3$ examples in Table~\ref{tab:examples} attain the bound at lift sizes below these sufficient thresholds.

\begin{table}[!ht]
\centering
\caption{Sufficient lift sizes for attaining $U=(J+1)!$. Each guarantee holds for every integer $P$ at least as large as the displayed value in its column.}
\label{tab:lift-thresholds}
\renewcommand{\arraystretch}{1.12}
\begin{tabular}{ccc rr}
\toprule
$J$ & $L$ & $U$ & Explicit construction & Probability lower bound $>0$ \\
 & & & $P_{\mathrm{exp}}$ & $P_{+}$ \\
\midrule
2 & 3 & 6 & $3.03\times10^{4}$ & $2.94\times10^{3}$ \\
2 & 4 & 6 & $1.12\times10^{6}$ & $9.25\times10^{3}$ \\
2 & 5 & 6 & $4.03\times10^{7}$ & $2.26\times10^{4}$ \\
2 & 6 & 6 & $1.46\times10^{9}$ & $4.68\times10^{4}$ \\
2 & 7 & 6 & $5.23\times10^{10}$ & $8.66\times10^{4}$ \\
2 & 8 & 6 & $1.89\times10^{12}$ & $1.48\times10^{5}$ \\
\midrule
3 & 4 & 24 & $3.35\times10^{16}$ & $1.01\times10^{45}$ \\
3 & 5 & 24 & $4.63\times10^{20}$ & $1.37\times10^{47}$ \\
3 & 6 & 24 & $6.40\times10^{24}$ & $7.53\times10^{48}$ \\
3 & 7 & 24 & $8.85\times10^{28}$ & $2.24\times10^{50}$ \\
3 & 8 & 24 & $1.23\times10^{33}$ & $4.22\times10^{51}$ \\
\midrule
4 & 5 & 120 & $3.77\times10^{41}$ & $4.19\times10^{471}$ \\
4 & 6 & 120 & $7.82\times10^{49}$ & $9.24\times10^{480}$ \\
4 & 7 & 120 & $1.63\times10^{58}$ & $7.33\times10^{488}$ \\
4 & 8 & 120 & $3.37\times10^{66}$ & $5.11\times10^{495}$ \\
\bottomrule
\end{tabular}
\end{table}
\FloatBarrier

\section{Conditions for an extension to quantum codes}
\label{sec:quantum}
Applying the distance guarantee to the CPM-PP quantum codes that motivated this study requires distinguishing nonzero codewords from nontrivial logical operators. We describe the restriction on directly using the check matrices in Table~\ref{tab:examples} and the remaining question for a general extension.

The binary check matrices $H_X,H_Z$ of a Calderbank--Shor--Steane (CSS) code satisfy $H_XH_Z^{\mathsf T}=0$. Assume the number of encoded qubits $k_Q=n-\rank_{\F_2}H_X-\rank_{\F_2}H_Z$ is positive. The $X$-type distance is
\begin{equation}
 d_X=\min\{\wt(x):H_Zx=0,\ x\notin\operatorname{row}(H_X)\},
 \label{eq:css-distance}
\end{equation}
Here $\operatorname{row}(H_X)$ is the binary row space; its elements represent stabilizers, which act trivially on the code space. Interchanging $X,Z$ defines $d_Z$, and the quantum minimum distance is $d_Q=\min(d_X,d_Z)$~\cite{okada2026pp}.

Let $H_Z=H_P(E)$ be a check matrix from Table~\ref{tab:examples}. Every nonzero row $r$ of a commuting matrix $H_X$ satisfies $H_Zr^{\mathsf T}=0$, so \eqref{eq:code} gives $\wt(r)\ge\dmin(C_P(E))$. For $J=3,L=4,P=24$, for example, every nonzero check on the other side has weight at least 24, excluding weight-4 rows. Every tabulated example has a distance lower bound greater than $L$, so none can form a CSS code with a check matrix having nonzero rows of the same weight $L$.

In a general CSS code, however, rows of $H_X$ belong to $\ker H_Z$ but are excluded from the minimization in \eqref{eq:css-distance}. The check weight $L$ therefore does not bound the quantum distance. Proving $d_X\ge U$ requires showing that every $x$ of weight below $U$ satisfying $H_Zx=0$ belongs to $\operatorname{row}(H_X)$. The present strategy of excluding every nonzero kernel vector cannot be used directly when stabilizers of weight below $U$ are present.

In a PP construction, linear congruences couple the $X$- and $Z$-exponents to ensure commutation~\cite{okada2026pp}. Applying \eqref{eq:construction} independently to both sides, or sampling all exponents independently and uniformly, therefore need not satisfy commutation. These constraints can also make low-weight cycle conditions hold identically, so Lemma~\ref{lem:generic}, which assumes independent indeterminates, does not apply directly. An extension requires showing that every word below the target distance that exists identically under the constraints is a sum of stabilizer rows, and then avoiding the remaining cycle conditions. Equality for the distance also requires proving that a word providing the upper bound is not a stabilizer. Obtaining this classification and construction for general admissible $J,L$ remains open; we have not established a quantum distance-attainment theorem with the generality of Theorem~\ref{thm:main}.

\bibliographystyle{IEEEtran}
\bibliography{references}
\end{document}